\documentclass[10pt, uplatex]{article}
\usepackage{color, amssymb, amsthm, amsmath, amsfonts, ascmac, comment, physics}
\usepackage{graphicx}
\usepackage{enumerate}
\usepackage{algorithm}
\usepackage{authblk}
\usepackage{algpseudocode}
\usepackage{cleveref}

\newtheorem{Def}{Definition}
\newtheorem*{Def*}{Definition}
\newtheorem{Thm}{Theorem}
\newtheorem{Prop}{Proposition}

\newtheorem{Lem}{Lemma}
\newtheorem{Cor}{Corollary}
\newtheorem{Fact}{Fact}
\newtheorem{Remark}{Remark}
\newtheorem{Algo}{Algorithm}
\newtheorem*{Thm1}{\rm\bf Theorem~\ref{Thm_tight_B}}
\newtheorem*{Prop1}{\rm\bf Proposition~\ref{Prop_poly_est_PAC}}

\crefname{Prop}{Proposition}{Propositions}
\crefname{Lem}{Lemma}{Lemmas}
\crefname{Fact}{Fact}{Facts}
\crefname{figure}{Fig}{Figures}
\crefname{Thm}{Theorem}{Theorems}
\crefname{Cor}{Corollary}{Corollaries}
\crefname{Def}{Definition}{Definitions}
\crefname{Algo}{Algorithm}{Algorithms}
\crefname{section}{Section}{Sections}

\renewcommand{\Pr}{\mathrm{Pr}}

\newcommand{\Bset}{\{0, 1\}}
\newcommand{\T}{\mathsf{T}}
\newcommand{\Mand}{\textrm{and}}
\newcommand{\E}{\mathbf{E}}
\newcommand{\Bias}{\mathbf{Bias}}
\newcommand{\Garb}{\mathsf{Garbage}}

\newcommand{\Init}{\textrm{init}}
\newcommand{\Sup}{\textrm{sup}}

\newcommand{\Range}{\mathrm{range}}
\newcommand{\Ghada}{\textsf{G-Hadamard}}

\newcommand{\Ud}{{\mathsf{U}(d)}}
\newcommand{\Sm}{{\mathfrak{S}_m}}
\newcommand{\REPf}{\mathsf{Rep}_{\varepsilon}(f)}
\newcommand{\SUPN}[1]{\|#1\|_{\sup}}
\newcommand{\LTN}[1]{\|#1\|_{L^2}}
\newcommand{\Ket}[1]{|#1\rangle}
\newcommand{\Est}[2]{\mathsf{Estimation}(#1, #2)}

\renewcommand{\Re}{\operatorname{Re}}
\renewcommand{\Im}{\operatorname{Im}}
\usepackage[top=30truemm,bottom=30truemm,left=20truemm,right=20truemm]{geometry}
\title{Another generalization of Hadamard test: Optimal sample complexities for learning functions on the unitary group}
\author{Daiki Suruga}
\affil{CFT PAN, Poland}
\begin{document}

\maketitle

\begin{abstract}
Estimating properties of unknown unitary operations is a fundamental task in quantum information science. 
While full unitary tomography requires a number of samples to the unknown unitary scaling linearly with the dimension (implying exponentially with the number of qubits), estimating specific functions of a unitary can be significantly more efficient. 
In this paper, we present a unified framework for the sample-efficient estimation of arbitrary square integrable functions $f: \Ud \to \mathbb{C}$, using only access to the controlled-unitary operation.
We first provide a tight characterization of the optimal sample complexity when the accuracy is measured by the averaged bias over the unitary $\Ud$. 
We then construct a sample-efficient estimation algorithm that becomes optimal under the Probably Approximately Correct (PAC) learning criterion for various classes of functions.

Applications include optimal estimation of matrix elements of irreducible representations, the trace, determinant, and general polynomial functions on $\Ud$. Our technique generalize the Hadamard test and leverage tools from representation theory, yielding both lower and upper bound on sample complexity.

\end{abstract}

\section{Introduction}

In quantum mechanics, any time evolution of a closed system is described by a unitary operator. 
As a result, many tasks in quantum information science involve estimating characteristics of a given, yet unknown unitary process. 
Examples of such tasks include the validation and certification of quantum circuits, error mitigation and correction in quantum processes, and the identification of hidden or unknown dynamics in black-box quantum systems. 
For such tasks, one needs to accurately estimate some characteristics of the unitaries, such as deviation from the ideal process (i.e., noise), eigenvalues and eigenvectors, or even the complete description of the unitary operator.
With the rapid theoretical and experimental advancement of quantum information technologies, developing efficient methods for estimating the properties of unknown unitary operations has become increasingly essential, and has therefore received much attention. 

A standard approach to estimating the characteristics of a given unitary operator is to reconstruct the entire unitary process—that is, to estimate all matrix elements of the unitary operator (sometimes up to a global phase). 
As a full characterization of a given unitary process is a fundamental task in quantum information science,  numerous studies have addressed the full unitary tomography problem. 
See, e.g., ~\cite{AJV01,CDS05,Hay06,YRC20,HKOT23}, as well as a simple method described in Nielsen and Chuang~\cite[Section~8.4.2]{CN10}. 
In particular, Ref.~\cite{HKOT23} recently showed that estimating a $d$-dimensional black-box unitary with high probability (when accuracy is measured by the diamond norm) requires $\Theta(d^2)$ samples. 
This implies that the number of required samples scales exponentially with the number of qubits.

On the other hand, many other studies focus on estimating a specific characteristic of a black-box unitary $U$; here the characteristic is formally defined as the value $f(U)$ of a function $f: \Ud \to \mathbb{C}$ (where $\Ud$ denotes the unitary group of dimension $d$).
For example, the (normalized) trace $\frac{1}{d} \Tr U$ of a unitary $U$ may be estimated by a simple application of the Hadamard test (see, e.g., ~\cite[Section~16.1]{Childs}) with only a constant number of queries to the controlled $U$ operation, no matter how many qubits the unitary $U$ acts on. 
This application of the Hadamard test provides an exponentially more sample-efficient method (with respect to the number of qubits) than a simple application of the full unitary tomography where the given unitary $U$ is directly estimated by the full tomography first, and then $\frac{1}{d}\Tr U$ is computed by classical computation.
As in this example, it is often the case that one provably reduces the number of samples significantly compared to the full unitary tomography, when focusing on a specific characteristic. 
Because of its provable efficiency, there are many works that focus on estimating a specific characteristic, including
the determinant $\det U$~\cite{ZBQ+25,AS25}, the eigenvalues~\cite{Kit95,DJSW07,SHF13,KLY15,MdW23}, and other quantities~\cite{LSS+20,SY23,FEK25}.

While there is a rich literature on estimating specific functions $f(U)$, only a few such functions, such as the trace and the determinant, have been systematically studied. For general functions $f: \Ud \to \mathbb{C}$, no efficient estimation scheme is currently known. It is therefore desirable to develop a unified, sample-efficient estimation framework that applies to a broad class of such functions.

\subsection{Our contributions}
The main purpose of this paper is to provide a unified estimation framework that allows for the construction of sample-efficient algorithms for estimating $f(U)$ for \emph{any choice} of the function $f(U)$.
 
Our results are twofold: 
\begin{enumerate}[(i)]
\item 
In our first result, we tightly characterize the sample complexity required to estimate a function $f$ accurately, when the accuracy is measured by \emph{the averaged bias}  over the Haar measure on $\Ud$. This works for any function $f$ that is square-integrable, i.e., $f \in L^2(\Ud)$.

\item In our second result, we provide an estimation algorithm based on the one designed for the first result. This algorithm become sample-optimal for estimating several kinds of functions even in the PAC (Probably Approximately Correct) learning framework, one of the most standard frameworks for tomography/estimation.

\end{enumerate}
In~\cref{subsub_First,subsub_Second}, each of the two results is explained in detail respectively.

\subsubsection{First result}\label{subsub_First}
To state our first result formally, we define the averaged bias $\Bias_G(\mathcal{A}, f)$ of an estimation algorithm $\mathcal{A}$ for a function $f$ as follows:

\begin{Def*}
Let $f: \Ud \to \mathbb{C}$.
We say a query algorithm $\mathcal{A}$ estimates the function $f$ correctly with the averaged bias over the group $G =\Ud$ less than $\varepsilon$ if and only if 

\begin{equation*}
\Bias_G(\mathcal{A}, f) := \E_{G}\left[ \Bias_g(\mathcal{A}, f)\right] < \varepsilon
\end{equation*}
where $\Bias_g(\mathcal{A}, f)$ is the bias of the algorithm $\mathcal{A}$ at $g \in \Ud$, and the expectation $\E_G$ is taken over the normalized Haar measure over the group $\Ud$.
\end{Def*}

One first result is~\cref{Thm_tight_B}, that shows the number of optimal query access for estimating $f$ under the condition $\Bias_G(\mathcal{A}, f) < \varepsilon$ is characterized by the quantity $\REPf$. The formal definition of the quantity $\REPf$ is defined later in~\cref{Def_REPf}.

\begin{Thm}\label{Thm_tight_B}
Let $f \in L^2(\Ud)$, and $\REPf$ be as in~\cref{Def_REPf}.
\begin{itemize}
\item For any query algorithm $\mathcal{A}$ that queries to \textrm{controlled}-g estimates $f$ with $\Bias_G(\mathcal{A}, f) < \varepsilon$,
then the algorithm $\mathcal{A}$ requires $\Omega(\REPf)$ query access. 
\item There is an estimation algorithm $\mathcal{A}$ with $O(\REPf)$ query access that satisfies $\Bias_G(\mathcal{A}, f) < \varepsilon$.
\end{itemize}
\end{Thm}
Together, these bounds imply that the optimal query complexity for estimation $f$ with accuracy $\varepsilon$ is
\begin{equation*}
\Theta(\REPf).
\end{equation*}

\paragraph{Remarks on~\cref{Thm_tight_B}.} We now give several remarks regarding~\cref{Thm_tight_B}.
\begin{itemize}
\item First, we remark that the measure of accuracy $\Bias_G$ is not the standard choice considered in the literature. 
The most commonly used measure is arguably the one defined in the PAC learning framework, which is in the scope of our second result.
Compared to the PAC-based criterion, our measure imposes in some sense a weaker requirement; the number of query access for accurately estimating $f(U)$ with respect to  the measure $\Bias_G$ is always smaller than with respect to  the PAC learning measure (under a mild condition) as shown in~\cref{Fact_B_Q}.
Nevertheless, the measure $\Bias_G$ is not merely an artificial construct, since the bias  is indeed the natural choice in classical/quantum estimation theory~\cite{CL98, Hel69,Hay16}, and there are research works, e.g., References~\cite{AJV01, CDPS04, CDS05, Hay06}, that naturally consider the average of cost functions over the unitary group when estimating some properties of an unknown unitary.

\item We also remark that our estimation algorithm assumes query access only to the controlled-unitary operation C-$U$, which cannot, in general, be replaced by access to $U$ alone, since different unitaries can produce the same unitary channels\footnote{Consider $U$ and $e^{i \theta}U$, even though recent papers~\cite{CCP+24, TW25} show query access to the controlled-U operation can be replaced by that of to the original $U$ operation under some conditions.}.
This contrasts with several other studies~\cite{Kit95, BHMT02, GSLW19, vACGN23} that allow their algorithms access to additional types of queries such as $U^*$, C-$U^*$, or multiple variants in combination with C-$U$.
Since in general it is hard~\cite{GST24} to construct these unitaries only from C-$U$, our algorithm has the desirable property of relying solely on accesses to C-$U$. 
\par
Furthermore, our lower bound holds even for algorithms that are permitted access to other types of queries such as $U$, $U^*$, C-$U^*$. Therefore our characterization remains tight for a wider class of query models.

\item Third, our result applies to a broad class of functions--specifically, the set of square-integrable functions $L^2(\Ud)$. 
This space includes all continuous functions as well as certain discontinuous functions. In particular, our result holds for any continuous functions, including natural examples such as $\Tr U$ and $\det U$, as well as for some discontinuous functions--for example,
\begin{equation*}
f(U)
=
\begin{cases}
1& \text{if $U \in A$},\\
0& \text{otherwise}\\
\end{cases}
\end{equation*}
where $A \subset \Ud$ is open or closed. These observations support the wide applicability of our result.

\item Lastly, we discuss the difficulty of computing the quantity $\REPf$. It is generally difficult to compute $\REPf$ accurately, and efficient computation of the quantity is beyond the scope of this paper. 
Nevertheless, for specific examples, we present methods to obtain simpler forms of $\REPf$ in~\cref{subsec_Simp}.
\end{itemize}

\subsubsection{Second result}\label{subsub_Second}
In our second result, we aim to estimate a function $f: \Ud \to \mathbb{C}$ accurately under the PAC criterion.
The definition of the PAC criterion is as follows:

\begin{Def*}
Let $f: \Ud \to \mathbb{C}$.
We say an algorithm $\mathcal{A}$ that has query access to C-$U$ estimates the function $f$ correctly with the precision parameters $(\varepsilon, \delta)$ if and only if 

\begin{equation*}
\forall g \in \Ud, \quad \Pr_S\left(|f(g) - \hat{f}(S)| > \varepsilon \right) < \delta.
\end{equation*}
\end{Def*}

Our second result, stated in~\cref{Prop_poly_est_PAC}, shows an estimation algorithm for a function $f(U)$ when $f(U)$ is a polynomial of $u_{ij}$'s and $\bar{u}_{ij}$'s with degree $\leq m$, where $U = (u_{ij})_{1 \leq i,j\leq d}$ and $\bar{u}_{ij}$ is the complex conjugate of $u_{ij}$.

\begin{Prop}\label{Prop_poly_est_PAC}
Let $f(U)$ be a polynomial of $u_{ij}$'s and $\bar{u}_{ij}$'s with degree $\leq m$.
Then there is an estimation algorithm for $f$ under the PAC criterion that uses
$O\left(\frac{\|A\|_1^2 \log \frac{1}{\delta}}{\varepsilon^2 }\cdot m\right)$ queries, for any matrix $A$ satisfying 
\begin{equation}\label{eq_Prop_poly_PAC}
f(g) =\Tr A\left(\bigoplus_{0 \leq n, n' \leq m} g^{\otimes n} \otimes g^{\ast \otimes n'} \otimes I_E\right).
\end{equation}
\end{Prop}

\cref{Prop_poly_est_PAC} shows that to get an upper bound on the query complexity for estimating $f$ under the PAC criterion, it is sufficient only to find a matrix $A$ that satisfies the constraint~\eqref{eq_Prop_poly_PAC} and has a small $L^1$ norm.
Based on~\cref{Prop_poly_est_PAC}, we give matching upper bounds for the (normalized) trace, the determinant, and other specific functions, in~\cref{subsec_Opt}.
All of the upper bounds given in~\cref{subsec_Opt} are tight when $\varepsilon$ and $\delta$ are sufficiently small constants.

\par
Notably, one of our applications of~\cref{Prop_poly_est_PAC} shows the optimal query complexity for every matrix element $\pi_{i,j}(g)$ of a unitary irreducible representation $\pi(g) =(\pi_{i, j}(g))$, a fundamental quantity in representation theory.

This quantity naturally appears, for example, in a generalization of Fourier analysis. In the generalization of Fourier analysis--specifically the harmonic analysis over compact groups--any function $f \in L^2(\Ud)$ is decomposed as a linear decomposition of the matrix elements of irreducible representations;
\begin{equation*}
f(g)
= \sum_{\pi, i, j}a_{i,j}^\pi \pi_{i,j}(g)
\end{equation*}
for some $a_{i,j}^\pi \in \mathbb{C}$. 
As representation theory provides powerful tools in quantum computing, there is a considerable number of works that focus on the computation of the matrix elements or its relevant quantities~\cite{Jor09, MS14,Cir24,BCG+24,BNZ25,LH25,PAN25,BGHS25}. 
For example, Ref.~\cite{Jor09} considers the gate complexity for estimating the matrix elements. However, the query complexity of the matrix elements was not known before this paper.

\subsection{Proof techniques}

We now briefly explain how to prove the lower bounds and the upper bounds in~\cref{Thm_tight_B} and~\cref{Prop_poly_est_PAC} respectively.
\paragraph{Lower bounds.}
To show the lower bound, let $\mathcal{A}$ be a query algorithm that uses $m$ queries for estimating a function $f$ with $\Bias_G < \varepsilon$.
First, we observe that the expectation $\E[\mathcal{A}|g]$ of the estimates of $\mathcal{A}$, when the unknown unitary is $g \in \Ud$, is a polynomial as the matrix elements of $g$ and $g^*$ with degree $\leq 2 m$.
Moreover, from the constraint on the bias: $\Bias_G < \varepsilon$, it follows that the $L^2$ distance $\|\E[{\mathcal{A}}|g] - f(g)\|_{L^2}$ is small.  This means the expectation $\E[\mathcal{A}|g]$ approximates the function $f$ well.
However several representation theory techniques tell the function is decomposed by 
\begin{equation*}
f
= \mathsf{poly}_{\leq 2m}f \oplus \mathsf{poly}_{\leq 2m}^\perp f \text{~satisfying~} \|f\|^2_{L^2} = \|\mathsf{poly}_{\leq 2m}f\|^2_{L^2} + \|\mathsf{poly}_{\leq 2m}^\perp f\|^2_{L^2},
\end{equation*}
where $\mathsf{poly}_{\leq 2m}f$ represents the polynomial with degree $\leq m$ that approximates the original function $f$ best, and $\mathsf{poly}_{\leq 2m}^\perp f$ is the rest.
Therefore if $\|\mathsf{poly}_{\leq 2m}^\perp f\|_{L^2}$ is large, the expectation $\E[\mathcal{A}|g]$ cannot approximate the original $f$ well.
This implies that the number $m$ of queries has to be large enough so that the quantity $\|\mathsf{poly}_{\leq 2m}^\perp f\|_{L^2}$ becomes sufficiently small, yielding a query lower bound.
\par
Similar methods sometimes appear in the literature, e.g.,~References~\cite{BBC+01, SY23}.
\paragraph{Upper bounds.}
To construct the query algorithm that attains our upper bounds, we generalize the Hadamard test and obtain an unbiased estimator for any polynomial $f(g)$. 
The standard Hadamard test provides an unbiased estimator for $\Re \Tr\rho U$ and $\Im \Tr\rho U$, thus also for $\Tr \rho U$ with one query access to C-$U$ for any state $\rho$. 
Generalizing this, we first develop an algorithm $\Ghada$ that provides an unbiased estimator for the inner product
\begin{equation*}
\langle \varphi |\left(\bigoplus_{0 \leq n, n' \leq m} g^{\otimes n} \otimes g^{\ast \otimes n'} \otimes I_E \right) \Ket{{\psi}}
\end{equation*}
for any pure states $\Ket{\varphi}, \Ket{\psi}$, with $2m$ query access to C-$g$ operation in~\cref{Fig_circuit_whole}.
We then observe that any polynomial function $f(g)$ with degree $\leq m$ may be represented as
\begin{equation*}
f(g)
= \Tr A_f \left(\bigoplus_{0 \leq n, n' \leq m} g^{\otimes n} \otimes g^{\ast \otimes n'} \otimes I_E \right)
\end{equation*}
for some matrix $A_f$.
The singular value decomposition on the matrix $A_f$ then tells
\begin{equation*}
f(g)
= \sum_{\substack{\sigma_i\\ \text{singular values of $A$}}} \sigma_i \langle \varphi_i|\left(\bigoplus_{0 \leq n, n' \leq m} g^{\otimes n} \otimes g^{\ast \otimes n'} \otimes I_E \right)\Ket{\psi_i},
\end{equation*}
and the $\Ghada$ is now used to estimate each of the terms indexed by $i$. This yields an unbiased estimator for a polynomial $f$ with degree $\leq m$.

Other generalizations of the Hadamard test has been investigated much in the literature e.g.,~\cite{WRZ+21, WBP+24, FEK25}.

\subsection{Organization of the paper}
The remainder of this paper is organized as follows. \cref{sec_Preli} introduces preliminary materials necessary for understanding the subsequent sections.  We then prove the lower bound in~\cref{sec_Lower} and the upper bound in~\cref{sec_Upper}. Their applications are discussed in~\cref{sec_Appli}, and some proofs are left to~\cref{sec_App}.

\section{Preliminaries}\label{sec_Preli}
Throughout this paper, the set of all $n \times m$ complex-valued matrices is denoted as $\mathbf{M}(n, m, \mathbb{C})$.
For a matrix $A \in \mathbf{M}(n, m, \mathbb{C})$, $\overline{A} \in \mathbf{M}(n, m, \mathbb{C})$ expresses the matrix whose entries are complex conjugate of the original matrix $A$, and $A^\ast \in \mathbf{M}(m, n, \mathbb{C})$ represents its adjoint matrix.
For any rectangular matrices $A = (a_{ij})$ and $ B$,  define
\begin{equation*}
A \oplus B
:= \mqty(A & \mathbf{0} \\ \mathbf{0} &B), \quad
A \otimes B
:= \mqty(
a_{11}B & a_{12}B & \cdots & a_{1n}B\\
a_{21}B & a_{22}B & \cdots & a_{2n}B\\
\vdots &  & \ddots & \vdots\\
a_{n1}B & a_{n2}B &  & a_{nn}B
).
\end{equation*}
Let $\mathbb{N}_0 = \{0, 1, \ldots \}$.
Let $\mathsf{U}(d)$ be the $d$-dimensional unitary group. 
Let $\textrm{C-}U$ be the controlled unitary operation, whose picture is expressed by the black bullet $\bullet$ on the controlled qubit throughout this paper,
 whereas the white bullet $\circ$ corresponds to $X\otimes I \cdot (\textrm{C-}U)\cdot X\otimes I$  that applies $U$ only when the controlled qubit is $|0\rangle$. (See~\cref{Fig_circuit_whole} for example.)
Let $H, X$, $S$ 
and Toffoli gate be the gates defined in the standard manner. See e.g.,~\cite{CN10}.

For a function $f: \Ud \to \mathbb{C}$, define the supremum norm
\begin{equation*}
\SUPN{f}
:= \sup_{g \in \Ud} |f(g)|
\end{equation*}
and an inner product
\begin{equation*}
\langle f, h\rangle 
:= \int_G \overline{f(g)} h(g) dg
\end{equation*}

for $f, h \in L^2(\Ud)$, where $dg$ is the normalized Haar measure on $\Ud$. This inner product induces the $L^2$ norm:
\begin{equation}\label{eq_Pre_inner-product}
\LTN{f}
:= \sqrt{\langle f, f\rangle }
\end{equation}

\subsection{Prerequisites from representation theory}
Here some standard materials in representation theory are described, which can be found in standard textbooks such as~\cite{FH13}.

\begin{Fact}\label{Fact_label_irrep_Ud}
There is a one-to-one correspondence between the set $\widehat{\Ud}$ of all irreducible representations over $\Ud$ and the set 
\begin{equation*}
\mathbb{Z}^d_+ := \{\lambda^d = (\lambda_1, \ldots, \lambda_d) \in \mathbb{Z}^d \mid \lambda_1 \geq \cdots \geq \lambda_d\}.
\end{equation*}
\end{Fact}

For an irreducible unitary representation $(\pi_\lambda, V_\lambda) \in \widehat{\Ud}$, let $\dim \pi_\lambda := \dim V_\lambda$, and $\pi_\lambda(g)_{i, j}~(1 \leq i, j \leq \dim \pi_\lambda)$ be the $(i, j)$ matrix entry of $\pi_\lambda(g)$.

\begin{Fact}\label{Fact_conjugate_label_Ud}
For an unitary representation $(\pi_\lambda, V_\lambda) \in \widehat{\Ud}$, its conjugate representation $(\bar{\pi}_\lambda, \bar{V}_\lambda)$ is also irreducible unitary representation and labeled by $\bar{\lambda} = (-\lambda_d, -\lambda_{d-1}, \ldots, -\lambda_1) \in \mathbb{Z}_+^d$.
\end{Fact}

\subsubsection{Peter--Weyl theorem and Schur--Weyl duality, and their consequences}
By the Peter--Weyl theorem~\cite{PW27}, the set of squared integrable functions over $\Ud$, denoted by $L^2(\Ud)$, is decomposed as

\begin{equation*}
L^2(\Ud)
= \widehat{\bigoplus}_{\lambda \in \mathbb{Z}_d^+} M_{\pi_\lambda}
\end{equation*}
(as the $\Ud$ representation,)
where $M_{\pi_\lambda}$ is the space spanned by the corresponding matrix components: 
\begin{equation*}
M_{\pi_\lambda} := \textrm{span}_{\mathbb{C}}\{\pi_\lambda(g)_{i, j} \mid 1 \leq i, j \leq \dim \pi_\lambda \},
\end{equation*}
and $\widehat{\bigoplus}$ represents a direct sum as Hilbert space. 
Each $\pi_\lambda(g)_{i,j}$ is orthogonal to others w.r.t. the inner product defined in~Equation~\eqref{eq_Pre_inner-product}. 
\begin{Fact}[Parseval type identity]\label{Fact_Parseval}
For any function $f \in L^2(\Ud)$, its $L^2$ norm satisfies
\begin{equation*}
\LTN{f}^2
= \sum_{\pi_\lambda \in \widehat{\Ud}} \LTN{P_\lambda f}^2
\end{equation*}
where $P_\lambda: L^2(\Ud) \to L^2(\Ud)$ is the orthogonal projection onto the subspace $M_{\pi_\lambda}$.
\end{Fact}

To state following facts in a concise manner, define
\begin{equation*}
\Lambda_{m, d}
:= \left\{\lambda \in \mathbb{Z}_+^d \mid \sum_{i \leq d} \lambda_i = m, \lambda_d \geq 0\right\}
\text{~and~}
\overline{\Lambda}_{m, d} 
:= \left\{\lambda \in \mathbb{Z}_+^d \mid \sum_{i \leq d} \lambda_i = -m, \lambda_1 \leq 0\right\}
\end{equation*}
for $d, m \geq 1$.

The following fact is one version of Schur--Weyl duality~\cite[Section~5.19]{EGH+11}.

\begin{Fact}\label{Fact_Schur-Weyl}
There exists a unitary operator $U_\mathsf{Sch} \in \mathbf{M}(d^m, d^m, \mathbb{C})$ such that 
\begin{equation*}
\forall g \in \Ud, ~~ U_\mathsf{Sch} g^{\otimes m} U^*_\mathsf{Sch} = \bigoplus_{\substack{\lambda \in \Lambda_{m,  d}}} I_\lambda \otimes \pi_\lambda(g) \otimes I_{m_{\lambda}}
\end{equation*}
where $m_\lambda$'s are positive integers.
\end{Fact}

\cref{Fact_Schur-Weyl}, known as Mixed Schur--Weyl duality~\cite{BCH+94, Ngu23, Gri25}, also play an important role for our purpose.
\begin{Fact}[Mixed Schur--Weyl duality]\label{Fact_Mixed-Schur-Weyl}
There exists a unitary operator $W_\mathsf{Sch}(m, \bar{m}) \in \mathbf{M}(d^{(m + \bar{m})}, d^{(m+ \bar{m})}, \mathbb{C})$ such that 
\begin{equation*}
\forall g \in \Ud, ~~ W_\mathsf{Sch}(m, \bar{m}) g^{\otimes m} \otimes \bar{g}^{\otimes \bar{m}} W^*_\mathsf{Sch}(m, \bar{m}) = \bigoplus_{\substack{\lambda \in \Lambda_d(m, \bar{m})}} \pi_\lambda(g) \otimes I_{m_\lambda}
\end{equation*}
for some positive $m_\lambda$'s, 
where 
\begin{equation*}
\Lambda_{d}(m, \bar{m})
= \left\{(\lambda_+, \lambda_-) \in \mathbb{Z}_+^d \mid 0 \leq d_m, d_{\bar{m}} \leq d,\quad 0 \leq k \leq \min\{m, \bar{m}\}, \quad \lambda_+ \in \Lambda_{m-k,d_m}, \quad \lambda_- \in \overline{\Lambda}_{\bar{m}-k,d_{\bar{m}}}\right\}.
\end{equation*}
Note that some zeros might be padded intermediately so that any element belongs to $\mathbb{Z}_+^d$; $(\lambda_+, \lambda_-) := (\lambda_+, 0, \ldots, 0, \lambda_-)$.
\end{Fact}

\subsection{Algorithms and Their Complexity}

As we deal with lower bounds on sample complexities, we next rigorously define the computational model considered throughout this paper in the below. 
\begin{Def}\label{Def_gene_q}
An $m$-generalized query algorithm $\mathcal{A}$ consists of a sequence of unitaries 

\begin{equation*}
U_0, V_1(g), U_1, V_2(g), U_2,\ldots, V_m(g), U_m
\end{equation*}
 where $V_i(g)$ is a unitary whose elements are polynomials in the elements of $g$ or $\bar{g}$ with degree at most one.
This algorithm has the initial register $|0\rangle^{\otimes K}$ where $K$ is a large constant, and finally performs the measurement on the state
\begin{equation*}
U_m V_m(g) \cdots U_2 V_2(g) U_1 V_1(g) U_0|0\rangle^{\otimes K}.
\end{equation*}
Without loss of generality, we assume the measurement is performed by the computational basis.
For our purpose, the algorithm for estimating a function $f: \Ud \to \mathbb{C}$ is additionally equipped with an estimator $\hat{f}: \Bset^K \to {\Range}(f)$. Based on the outcome $s \in \Bset^K$ of the final measurement, the algorithm outputs $\hat{f}(s)$ as an estimate.
\end{Def}

\begin{Remark}
This model is a bit more powerful than the ordinary computational models since
 we can take $g, g^\ast, \textrm{C-}g, \textrm{C-}g^\ast$, for example, as $V_i(g)$'s. This means our lower bounds hold in a wider class of computational models. Note that in our upper bounds we only use the ordinary computational model, which have access only to the controlled-$g$ operation, and therefore does not rely on any additional power possibly derived from~\cref{Def_gene_q}.
\end{Remark}

\begin{Def}
Let $f: \Ud \to \mathbb{C}$.
We say an $m$-generalized query algorithm $\mathcal{A}$ estimates the function $f$ correctly with the averaged bias over the group $G =\Ud$ less than $\varepsilon$ if and only if

\begin{equation*}
\Bias_G(\mathcal{A}, f) := \E_{G}\left[ \left|\E_S[\hat{f}(s)|g] - f(g)\right|^2\right] < \varepsilon
\end{equation*}
where $\E_S[\hat{f}(s)|g]:= \sum_{s \in S} |\langle s|\mathcal{A}(g)|0\rangle^{\otimes K}|^2 \hat{f}(s)$.
\end{Def}

\begin{Def}
Let $f: \Ud \to \mathbb{C}$.
We say an $m$-generalized query algorithm $\mathcal{A}$ estimates the function $f$ correctly with the precision parameters $(\varepsilon, \delta)$ iff

\begin{equation*}
\forall g \in \Ud, \quad \Pr_S\left(|f(g) - \hat{f}(S)| > \varepsilon \right) < \delta.
\end{equation*}
\end{Def}

\begin{Def}
Let $B_{\varepsilon}(f)$ be the minimum number of generalized queries required to estimate $f$ correctly with the averaged bias (over the group $\Ud$) less  than $\varepsilon$.
\end{Def}

\begin{Def}
Let $Q_{\varepsilon,\delta}(f)$ be the minimum number of generalized queries required to estimate $f$ correctly with the precision parameters $(\varepsilon, \delta)$.
\end{Def}

These two complexities are related as follows.

\begin{Fact}\label{Fact_B_Q}
For any function $f \in L^2(\Ud)$, $B_{\varepsilon'}(f) \leq Q_{\varepsilon, \delta}(f)$ where $\varepsilon' := (2 \delta \cdot \|f \|_\Sup  + \varepsilon)^2$.
\end{Fact}
\begin{proof}
For any generalized query algorithm $\mathcal{A}$ that has an estimator $\hat{f}$ for the function $f$,

\begin{align*}
\left|\E[\hat{f}(s)| g]  - f(g)\right|
&\leq \sum_{s: |\hat{f}(s) - f(g)| > \varepsilon } \Pr(s|g) |\hat{f}(s) - f(g)| \\
&+ \sum_{s: |\hat{f}(s) - f(g) \leq  \varepsilon } \Pr(s|g) |\hat{f}(s) - f(g)|
\end{align*}
holds for any $g$. Therefore we have
\begin{align*}
\left|\E[\hat{f}(s)| g]  - f(g)\right|
&\leq \sum_{s: |\hat{f}(s) - f(g)| > \varepsilon } \Pr(s|g) |\hat{f}(s) - f(g)| \\
&+ \sum_{s: |\hat{f}(s) - f(g)| \leq \varepsilon } \Pr(s|g) |\hat{f}(s) - f(g)|\\
&\leq 2\delta\|f\|_\Sup + \varepsilon
\end{align*}
when the algorithm $\mathcal{A}$ satisfies the PAC-learning condition: $\Pr(|\hat{f}(s) - f(g)| > \varepsilon) < \delta$ for any $g$.
\end{proof}

\section{Lower bound}\label{sec_Lower}
First, define the subspace of $L^2(\Ud)$ as follows.
\begin{equation}\label{eq_def_Q}
\mathcal{Q}_{ \leq m}(\Ud) := 
{\mathrm{span}}_\mathbb{C}\left\{g^{x_{00}}_{00}\bar{g}^{y_{00}}_{00}g^{x_{01}}_{01}\bar{g}^{y_{01}}_{01}\cdots g^{x_{dd}}_{dd}\bar{g}^{y_{dd}}_{dd} \mid \sum_{1 \leq i, j \leq d}x_{ij} + y_{ij} \leq m, ~~ (x_{ij}, y_{ij}) \in \mathbb{N}_0^2\right\}
\end{equation}
where $g_{ij}$'s are the matrix elements of a unitary $g \in \Ud$.
In other words, the space $\mathcal{Q}_{ \leq m}(\Ud)$ is the set of all polynomials with degree at most $m$.

The space $\mathcal{Q}_{\leq m}(\Ud)$ is characterized as follows.

\begin{Prop}\label{Prop_poly_m_projection}
\begin{equation}\label{eq_Prop_poly}
\mathcal{Q}_{\leq m}(\Ud)
= \bigoplus_{\substack{\lambda \in \Lambda_{d}(n, \bar{n})\\ 0 \leq n, \bar{n} \leq m}} M_{\pi_\lambda}.
\end{equation}
\end{Prop}

\begin{Remark}
\cref{Prop_poly_m_projection} implies that there is an orthogonal projection operator $Q_{\leq m}$ onto $\mathcal{Q}_{\leq m}(\Ud)$, that is the sum of the projections $P_\lambda$ where $\lambda \in \Lambda_{d}(n, \bar{n})~(0 \leq n, \bar{n} \leq m)$.
As the Peter--Weyl theorem tells, there is another projection operator $Q^\perp_{\leq m}$ onto the orthogonal subspace of $\mathcal{Q}_{\leq m}(\Ud)$.
By~\cref{Fact_Parseval}, these projections satisfy
\begin{equation*}
\|Q_{\leq m}f\|^2_{L^2}
+
\|Q_{\leq m}^\perp f\|^2_{L^2}
=
\|f\|^2_{L^2}
\end{equation*}
for any $f \in L^2(\Ud)$.
\end{Remark}

\begin{proof}
Applying~\cref{Fact_Mixed-Schur-Weyl} for any $0 \leq n, \bar{n} \leq m$, we observe that there exists a unitary matrix $W$ such that for any $g \in \Ud$,
\begin{equation}\label{eq_Prop_poly_proof}
\bigoplus_{0 \leq n, \bar{n} \leq m} g^{\otimes n} \otimes \bar{g}^{\otimes \bar{n}}
= W \bigoplus_{\substack{ 0 \leq n, \bar{n} \leq m\\ \lambda \in \Lambda_{d}(n, \bar{n})}} \pi_\lambda(g) \otimes I_{m_\lambda} W^\ast
\end{equation}
for some $m_\lambda$'s.
This shows that any matrix element in
$\bigoplus_{0 \leq n, \bar{n} \leq m} g^{\otimes n} \otimes \bar{g}^{\otimes \bar{n}}$ is expressed by a linear combination of $\pi_\lambda(g)_{i, j}$'s appeared in the RHS of Equation~\eqref{eq_Prop_poly_proof} and vice versa.
This shows that the RHS and LHS in~\eqref{eq_Prop_poly}
have the same elements, since every basis in $\mathcal{Q}_{\leq m}(\Ud)$ appears as a matrix element of  
$\bigoplus_{0 \leq n, \bar{n} \leq m} g^{\otimes n} \otimes \bar{g}^{\otimes \bar{n}}$. 
To show the direct sum property of the RHS, simply use the Peter--Weyl theorem. Therefore we obtain the desired statement.
\end{proof}

Let us next define the quantity $\REPf$ that plays a pivotal role in this paper.
\begin{Def}\label{Def_REPf}
For a function $f \in L^2(\Ud)$, define
\begin{equation*}
\REPf
:= \max\left\{m \in \mathbb{N} \mid \| {Q^\perp_{\leq 2m}f} \|^2_{L^2} \geq \varepsilon \right\}.
\end{equation*}
\end{Def}

\begin{Prop}\label{Prop_Lower_B}
For any function $f \in L^2(\Ud)$,  $B_{\varepsilon}(f) =  \Omega(\REPf).$
\end{Prop}
\begin{proof}

Let $\mathcal{A}$ be an $m$-generalized query algorithm equipped with an estimator $\hat{f}: \Bset^K \to {\Range} f$ for the function $f$ and assume $\mathcal{A}$ estimates $f$ with $\Bias_G < \varepsilon$;

\begin{equation*}
\Bias_G(\mathcal{A}, f) = \E_{G}[|\E_S[\hat{f}(s)|g] - f(g)|^2] < \varepsilon.
\end{equation*}
Since
$\E_S[\hat{f}(s)|g]:= \sum_{s \in S} |\langle s|\mathcal{A}(g)|0\rangle^{\otimes K}|^2 \hat{f}(s) \in \mathcal{Q}_{2m}(\Ud)$,
i.e., $\E_S[\hat{f}(s)|g]$ is polynomial (in the matrix elements of $\Ud$) with degree at most $2m$, we have $Q^\perp_{\leq 2m} (f - \E_S[\hat{f}(s)|g]) = Q^\perp_{\leq 2m} f$.
This implies
\begin{equation*}
\|Q^\perp_{\leq 2m} f\|^2_{L^2} = \|Q^\perp_{\leq 2m} (f - \E_S[\hat{f}(s)|g]) \|_{L^2}^2
 \leq \|f(g) - \E_S[\hat{f}(s)|g] \|_{L^2}^2 = \Bias(\mathcal{A}, f) < \varepsilon
\end{equation*}
where the first inequality follows from 
\begin{equation*}
\|Q^\perp_{\leq 2m}g\|^2_{L^2} = \|g\|^2_{L^2} - \|Q_{\leq 2m}g\|^2_{L^2} \leq \|g\|^2_{L^2}
\end{equation*}
for any $g \in L^2(\Ud)$.
\par
These arguments show that if $B_\varepsilon(f) \leq m$ then $\|Q_{\leq 2m}^\perp f\|_{L^2}^2 < \varepsilon$ holds.
Taking the contraposition of this statement shows the desired statement.
\end{proof}

\section{Upper bound}\label{sec_Upper}

\begin{figure}[hbtp]
 \centering
 \includegraphics[keepaspectratio, scale=0.45]
      {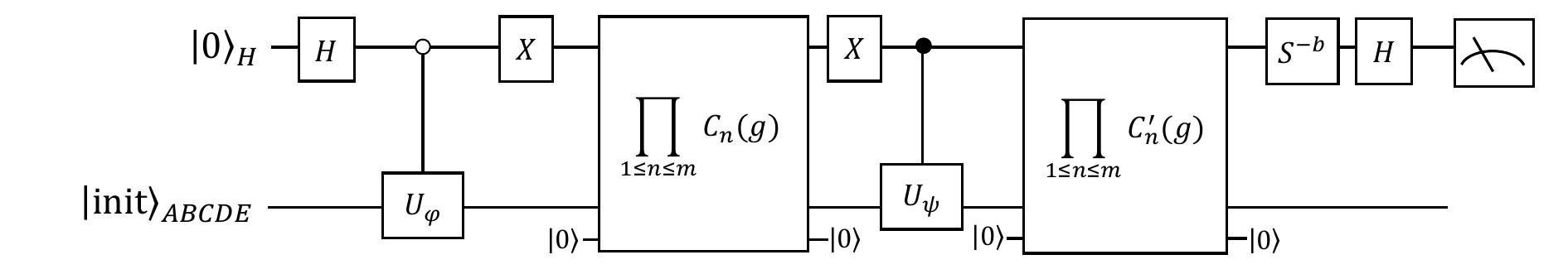}
 \caption{The definition of the algorithm \Ghada.}
 \label{Fig_circuit_whole}
\end{figure}

\begin{figure}[hbtp]
    \begin{tabular}{cc}
      \begin{minipage}[t]{0.45\hsize}
        \centering
        \includegraphics[keepaspectratio, scale=0.45]
		{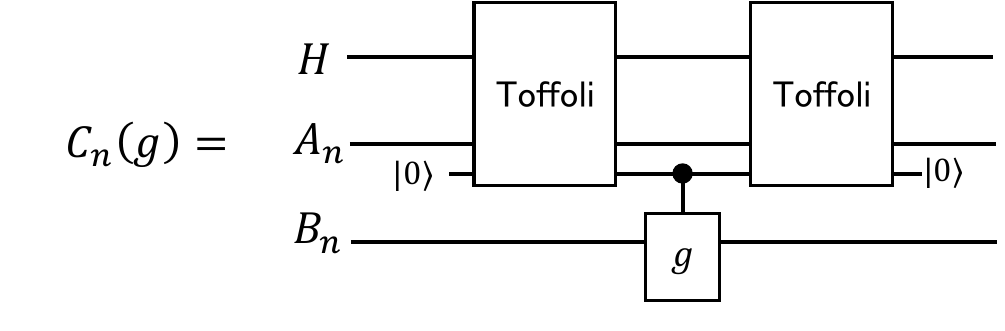}
      \end{minipage} &
      \begin{minipage}[t]{0.45\hsize}
        \centering
        \includegraphics[keepaspectratio, scale=0.45]
		{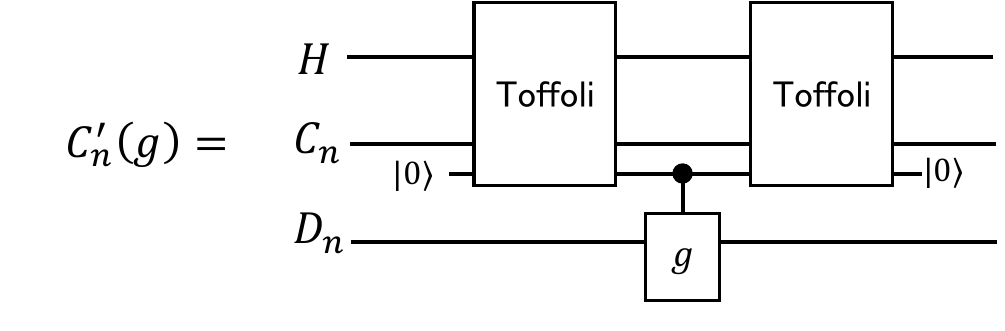}
      \end{minipage}
    \end{tabular}
 \caption{The definition of the components $C_n(g)$ and $C'_n(g)$.}
 \label{Fig_circuit_SWAP}
  \end{figure}

\begin{Algo}[Generalized Hadamard test: \Ghada]\label{Algo_G-Hadamard}
For any positive integers $m$ and $d$, a binary bit $b \in \Bset$, let $A_n, B_n, C_n$ and $D_n~~(1 \leq n \leq m)$ be quantum systems whose dimensions satisfy
\begin{equation*}
\dim A_n =\dim C_n =2, \quad
\dim B_n =\dim D_n = d,
\end{equation*}
and define 
\begin{equation*}
A := \bigotimes_{1 \leq n \leq m}A_n,~~
B := \bigotimes_{1 \leq n \leq m}B_n,~~
C := \bigotimes_{1 \leq n \leq m}C_n,~\Mand~
D := \bigotimes_{1 \leq n \leq m}D_n.
\end{equation*}
Additionally, for any $\varphi, \psi \in \mathbb{C}^{(2m)^d}$, let $U_\varphi$ and $U_\psi$ be unitary operators that satisfy $U_\varphi|\Init \rangle_{ABCDE} = \Ket{\varphi}$
and $U_\psi|\Init \rangle_{ABCDE} = \Ket{\psi}$ respectively where $E$ is a finite dimensional quantum system, and $|\Init\rangle$ is some initial state on $ABCDE$.

Under these definitions and notations, a new algorithm \Ghada,  is defined as in~\cref{Fig_circuit_whole} where $C_n(g)$ and $C_n'(g)$ are defined in~\cref{Fig_circuit_SWAP}, and the final measurement is performed with the computational basis on the first qubit system.
\end{Algo}

\begin{Prop}\label{}
The algorithm \Ghada~ uses $2 m$ queries to the controlled-$g$ operation. The probability $P_{0|b}~(b \in \Bset)$ of obtaining zero as the measurement outcome is 
\begin{equation*}
P_{0 |b = 0}
= \frac{1 + \Re \langle \varphi |(I \oplus g^*)_{AB}^{\otimes m} \otimes (I \oplus g)_{CD}^{\otimes m} \otimes I_E \Ket{\psi}}{2}
\end{equation*}
when $b$ is set to zero, and  
\begin{equation*}
P_{0 |b = 1}
= \frac{1 + \Im \langle \varphi |(I \oplus g^*)_{AB}^{\otimes m} \otimes (I \oplus g)_{CD}^{\otimes m} \otimes I_E \Ket{\psi}}{2}
\end{equation*}
when $b$ is set to one.
\footnote{Note $(I \oplus g)_{AB}^{\otimes m} = \bigotimes_{1 \leq n \leq m} (I \oplus g)_{A_n B_n}$ and $(I \oplus g)_{CD}^{\otimes m} = \bigotimes_{1 \leq n \leq m} (I \oplus g)_{C_n D_n}$ to be more precise.}
\end{Prop}

\begin{proof}
The number of the controlled-$g$ in \Ghada~ is obviously $2m$.
\par
In addition, the calculation of the desired probability is not so difficult.
First let us investigate how $C_n(g)$ and $C_n'(g)$ work on $HA_n B_n$ and $H C_n D_n$ respectively.
By definition, $C_n(g)$ acts as the identity on $|0_H, b_{A_n}, j_{B_n}\rangle$, and acts as 
$|1_H, b_{A_n}, j_{B_n}\rangle \mapsto |1_H\rangle (I \oplus g)| b_{A_n}, j_{B_n}\rangle$,
where $\{|b_{A_n}\rangle\}_{b \in \Bset}$ is the computational basis on the system $A_n$, and 
$\{|j_{B_n}\rangle\}_{j \leq d}$ is the computational basis on the system $B_n$. 
This means $C_n(g)$ acts in the same manner as the controlled-controlled-$g$ operation on $HA_n B_n$.
Similarly,  $C'_n(g)$ acts the same as the controlled-controlled-$g$ operation on $H C_n D_n$.

\par
In the rest of proof, we directly calculate the evolution of the initial state step by step.
By a simple calculation, the state right before applying $\Pi_{1 \leq n \leq m} C_n(g)$ in~\cref{Fig_circuit_whole} is
\begin{equation}\label{Eq_algo_step1}
\frac{1}{\sqrt{2}}\Ket{1} \otimes \Ket{\varphi}_{ABCDE}
+ \frac{1}{\sqrt{2}}\Ket{0} \otimes \Ket{\Init}_{ABCDE}
\end{equation}
Since $C_n(g)$'s act as the controlled-controlled-$g$ operation on $H A_n B_n$, after $\Pi_{1 \leq n \leq m} C_n(g)$ is applied, the state~\eqref{Eq_algo_step1}  becomes
\begin{equation}\label{Eq_algo_step2}
\frac{1}{\sqrt{2}}\Ket{1} \otimes (I \oplus g)^{\otimes m}_{AB} \Ket{\varphi}_{ABCDE}
+ \frac{1}{\sqrt{2}}\Ket{0} \otimes \Ket{\Init}_{ABCDE},
\end{equation}
Note that the controlled-$g$ operation on $A_n B_n$ is identical to $(I \oplus g)_{A_n B_n}$.
As in~\cref{Fig_circuit_whole}, we apply $X$ and the controlled-$U_\psi$ operation to the state~\eqref{Eq_algo_step2}, which yields 
\begin{equation}\label{Eq_algo_step3}
\frac{1}{\sqrt{2}}\Ket{0} \otimes (I \oplus g)^{\otimes m}_{AB} \Ket{\varphi}_{ABCDE}
+ \frac{1}{\sqrt{2}}\Ket{1} \otimes \Ket{\psi}_{ABCDE}.
\end{equation}
Similar to the case~\eqref{Eq_algo_step3}, the state then changes to 
\begin{equation}\label{Eq_algo_step4}
\frac{1}{\sqrt{2}}\Ket{0} \otimes (I \oplus g)^{\otimes m}_{AB} \Ket{\varphi}_{ABCDE}
+ \frac{1}{\sqrt{2}}\Ket{1} \otimes(I \oplus g)^{\otimes m}_{CD} \Ket{\psi}_{ABCDE}.
\end{equation}
by applying $\Pi_{1 \leq n \leq m} C_n'(g)$. Finally, $S^{-b}$ and $H$ are applied to the state~\eqref{Eq_algo_step4} which yields
\begin{align*}
&\frac{1}{2}\Ket{0} \otimes \left((I \oplus g)^{\otimes m}_{AB} \Ket{\varphi}_{ABCDE}  + (-i)^{b} (I \oplus g)^{\otimes m}_{CD} \Ket{\psi}_{ABCDE} \right)\\
&+ \frac{1}{2}\Ket{1} \otimes \left((I \oplus g)^{\otimes m}_{AB} \Ket{\varphi}_{ABCDE}  - (-i)^{b} (I \oplus g)^{\otimes m}_{CD} \Ket{\psi}_{ABCDE} \right).
\end{align*}
Therefore the final measurement produces zero with probability
\begin{equation*}
\frac{1}{2} \left((I \oplus g)^{\otimes m}_{AB} \Ket{\varphi}_{ABCDE}  + (-i)^{b} (I \oplus g)^{\otimes m}_{CD} \Ket{\psi}_{ABCDE} \right)^*
\frac{1}{2} \left((I \oplus g)^{\otimes m}_{AB} \Ket{\varphi}_{ABCDE}  + (-i)^{b} (I \oplus g)^{\otimes m}_{CD} \Ket{\psi}_{ABCDE} \right)
\end{equation*}
that equals to 
\begin{equation*}
P_{0 |b = 0}
= \frac{1 + \Re \langle \varphi |(I \oplus g^*)_{AB}^{\otimes m} \otimes (I \oplus g)_{CD}^{\otimes m} \otimes I_E \Ket{\psi}}{2}
\end{equation*}
when $b$ is set to zero, and  
\begin{equation*}
P_{0 |b = 1}
= \frac{1 + \Im \langle \varphi |(I \oplus g^*)_{AB}^{\otimes m} \otimes (I \oplus g)_{CD}^{\otimes m} \otimes I_E \Ket{\psi}}{2}
\end{equation*}
when $b$ is set to one.
This completes proof.
\end{proof}

\begin{Prop}\label{Prop_B_upper}
For any $f \in L^2(\Ud)$, $B_\varepsilon(f) = O(\REPf)$.
\end{Prop}

\begin{proof}
Our proof consists of two parts. In the first part, we create an estimator for
the inner product 
\begin{equation*}
\langle \tilde{\varphi} |\left(\bigoplus_{0 \leq n, n' \leq m} g^{\otimes n} \otimes g^{\ast \otimes n'} \otimes I_E \right) \Ket{\tilde{\psi}}
\end{equation*}

when a system $E$, and  $\tilde{\varphi}, \tilde{\psi}$ are given.
In the second part, based on the estimator, we propose an estimation process for any $f \in L^2(\Ud)$ and analyze its efficiency.

\paragraph{Proof of the first part.}
First, observe $(A \oplus B) \otimes (C \oplus D) = A \otimes(C \oplus D)  \oplus B \otimes(C \oplus D)$ and $A \otimes(B \oplus C) \simeq (A \otimes B) \oplus (A \otimes C) $ up to unitary equivalence, for any rectangular matrices $A, B, C$ and $D$.
This observation leads to 
\begin{equation*}
(I \oplus g)^{\otimes m} \simeq
\left( \bigoplus_{0 \leq n \leq m} g^{\otimes m}\right) \oplus \Garb,  ~~\Mand~~
(I \oplus g^\ast)^{\otimes m} \simeq
\left(\bigoplus_{0 \leq n \leq m} g^{\ast \otimes m}\right) \oplus \Garb
\end{equation*}
where $\Garb$ is some unitary, and $g^{\otimes 0} = I_d$.
Therefore
\begin{equation*}
(I \oplus g)^{\otimes m}_{AB}
\otimes
(I \oplus g^\ast)^{\otimes m}_{CD}
\simeq
\left(\bigoplus_{0 \leq n, n' \leq m} g^{\otimes n} \otimes g^{\ast \otimes n'} \right) \oplus \Garb
\end{equation*}
which further leads to 
\begin{equation*}
(I \oplus g)^{\otimes m}_{AB}
\otimes
(I \oplus g^\ast)^{\otimes m}_{CD} \otimes I_E
\simeq
\left(\bigoplus_{0 \leq n, n' \leq m} g^{\otimes n} \otimes g^{\ast \otimes n'} \right) \otimes I_E \oplus \Garb
\simeq
\left(\bigoplus_{0 \leq n, n' \leq m} g^{\otimes n} \otimes g^{\ast \otimes n'} \otimes I_E \right) \oplus \Garb
\end{equation*}
In other words, there is a unitary $W$ such that
\begin{equation*}
W^\ast \cdot (I \oplus g)^{\otimes m}_{AB}
\otimes
(I \oplus g^\ast)^{\otimes m}_{CD} \otimes I_E \cdot W
=
\left(\bigoplus_{0 \leq n, n' \leq m} g^{\otimes n} \otimes g^{\ast \otimes n'} \otimes I_E \right) \oplus \Garb.
\end{equation*}

\par
Now for any vectors $\tilde{\varphi}, \tilde{\psi}$ that have the same dimension as $\left(\bigoplus_{0 \leq n, n' \leq m} g^{\otimes n} \otimes g^{\ast \otimes n'} \otimes I_E\right)$, 
we implement~\cref{Algo_G-Hadamard}, \Ghada, with $\varphi = W (\tilde{\varphi} \oplus \mathbf{0}_\Garb)$ and $\psi = W(\tilde{\psi} \oplus \mathbf{0}_\Garb)$.
Then the probability of obtaining zero at the final measurement is
\begin{equation*}
P_{0 |b = 0}
= \frac{1 + \Re \langle \tilde{\varphi} |\left(\bigoplus_{0 \leq n, n' \leq m} g^{\otimes n} \otimes g^{\ast \otimes n'} \otimes I_E\right) \Ket{\tilde{\psi}}}{2}
\end{equation*}
when $b$ is set to zero, and  
\begin{equation*}
P_{0 |b = 1}
= \frac{1 + \Im \langle \tilde{\varphi} |\left(\bigoplus_{0 \leq n, n' \leq m} g^{\otimes n} \otimes g^{\ast \otimes n'} \otimes I_E \right) \Ket{\tilde{\psi}}}{2}
\end{equation*}
when $b$ is set to one.
As shown below, this implementation yields an unbiased estimator for 

\begin{equation*}
\langle \tilde{\varphi} |\left(\bigoplus_{0 \leq n, n' \leq m} g^{\otimes n} \otimes g^{\ast \otimes n'} \otimes I_E \right) \Ket{\tilde{\psi}}
\end{equation*}
with $2m$ controlled-$g$ operations. 

\par

To see its unbiasedness, let $b \in_U \Bset$ behave as the unbiased coin and perform~\cref{Algo_G-Hadamard} according to the value of $b$. Denote its measurement outcome by $M \in \Bset$ and then define the estimator $\hat{f}:(M, b) \in \Bset^2 \mapsto \mathbb{C}$ as
\begin{equation*}
\hat{f}(0, 0) = 4 - (1 + i), \quad
\hat{f}(0, 1) = 4i - (1 + i),\quad
\hat{f}(1, 1) = \hat{f}(1, 0) = -(1+i).
\end{equation*}
Then
\begin{align*}
\E[\hat{f}(M, b) + (1 + i)]&= \Pr(0, 0) \hat{f}(0, 0) + \Pr(0, 1)\hat{f}(0, 1) = 
\langle \tilde{\varphi} |\left(\bigoplus_{0 \leq n, n' \leq m} g^{\otimes n} \otimes g^{\ast \otimes n'} \otimes I_E \right) \Ket{\tilde{\psi}} + (1 + i).
\end{align*}
and therefore
\begin{equation*}
\E[\hat{f}(M, b)] = \langle \tilde{\varphi} |\left(\bigoplus_{0 \leq n, n' \leq m} g^{\otimes n} \otimes g^{\ast \otimes n'} \otimes I_E \right) \Ket{\tilde{\psi}}.
\end{equation*}

Denote this estimation process by $\Est{\tilde{\varphi}}{\tilde{\psi}}$.

\paragraph{Proof of the second part.}
Based on the above process $\Est{\tilde{\varphi}}{\tilde{\psi}}$, the second part of proof shows how to estimate a function $f \in L^2(\Ud)$ via polynomial approximation.
For any $f \in L^2(\Ud)$, let $m_0 := (\REPf +1)$ and observe that the projection of $f$ onto the space $\mathcal{Q}_{ \leq 2m_0}(\Ud)$ defined in Equation~\eqref{eq_def_Q} may be written as
\begin{equation}\label{Eq_proj_f}
Q_{\leq 2 m_0}(f)
= \sum_{\substack{\alpha, \alpha' \in \mathbb{N}_0^{d \times d}\\ |\alpha| + |\alpha'| \leq 2 m_0}} a(\alpha, \alpha') g^\alpha \bar{g}^{\alpha'}
\end{equation}
where $|\alpha| := \sum_{1 \leq i, j \leq d} \alpha_{ij}$ for $\alpha = (\alpha_{i, j})_{1 \leq i, j \leq d} \in \mathbb{N}_0^{d \times d}$,  $a(\alpha, \alpha') \in \mathbb{C}$, 
$g^\alpha := g_{11}^{\alpha_{11}} g_{12}^{\alpha_{12}} \cdots g_{dd}^{\alpha_{dd}}$, 
$\bar{g}^{\alpha'} := \bar{g}_{11}^{\alpha'_{11}} \bar{g}_{12}^{\alpha'_{12}} \cdots \bar{g}_{dd}^{\alpha'_{dd}}$.
\par
In the following, we aim to estimate $Q_{\leq 2 m_0}f$ instead of $f$ itself.
To this end, first observe that any $g^\alpha \bar{g}^{\alpha'}$ appears as a matrix element of $ \bigoplus_{0 \leq n, n' \leq 2m_0} g^{\otimes n} \otimes g^{\ast \otimes n'}$
that follows from the definition of tensor product of matrices.
In addition, let us observe that any function $f(x_{11}, x_{12}, \ldots, x_{dd})$ of the form $f = \sum_{i,j} a_{ij} x_{ij}$ may be expressed as $f =\Tr A X$ by a matrix $A$ when $X = (x_{ij})$.
These observations imply that
 there exists a matrix $A$ such that
\begin{equation*}
 \sum_{\substack{\alpha, \alpha' \in \mathbb{N}_0^{d \times d}\\ |\alpha| + |\alpha'| \leq 2 m_0}} a(\alpha, \alpha') g^\alpha \bar{g}^{\alpha'}
= \Tr \left[ A \bigoplus_{0 \leq n, n' \leq 2m_0} g^{\otimes n} \otimes g^{\ast \otimes n'} \right].
\end{equation*}
Together with Equation~\eqref{Eq_proj_f},
this equation further implies 
\begin{equation*}
Q_{\leq 2 m_0}(f)
= \Tr \left[ U D V \bigoplus_{0 \leq n, n' \leq 2m_0} g^{\otimes n} \otimes g^{\ast \otimes n'} \right]
= \sum_i \sigma_i \langle e_i | V \bigoplus_{0 \leq n, n' \leq 2m_0} g^{\otimes n} \otimes g^{\ast \otimes n'} U \Ket{e_i} 
\end{equation*}
by the singular value decomposition $A = U D V$, where $\sigma_i$'s are the singular values of $A$ and $e_i$'s are the standard basis.
This establishes
\begin{equation}\label{Eq_proj_est}
Q_{\leq 2 m_0}(f) = \|A\|_1 \cdot \sum_i \frac{\sigma_i}{\|A\|_1} \langle e_i | \left( V \bigoplus_{0 \leq n, n' \leq 2m_0} g^{\otimes n} \otimes g^{\otimes n'} U \right) \Ket{e_i}.
\end{equation}

We can now describe our estimation scheme for $Q_{\leq 2 m_0}f$ in a simple manner:
In our scheme, we first randomly choose some coordinate $i$ with probability $\sigma_i / \|A\|_1$, and then estimate 
$\langle e_i | \left( V \bigoplus_{0 \leq n, n' \leq 2m_0} g^{\otimes n} \otimes g^{\ast \otimes n'} U \right) \Ket{e_i}$ by $\Est{V^\ast e_i}{U e_i}$ with $E$ satisfying $\dim E = 1$, and finally multiply the value by $\|A\|_1$.
This scheme uses $4 m_0$ controlled-$g$ operations, and satisfies unibiasedness property for $Q_{ \leq 2 m_0}f$ due to the expression~\eqref{Eq_proj_est} and unbiasedness property of $\Est{\varphi}{\psi}$.

\par
We now investigate the bias of our estimation scheme for the original function $f$.
At each point $g \in \Ud$, the bias of our estimation scheme is 
\begin{equation*}
|Q_{ \leq 2 m_0}f(g) - f(g)|^2
\end{equation*}
and therefore the quantity $\Bias_G$, the average of the bias over the group $\Ud$ (with respect to the Haar measure), satisfies
\begin{equation*}
\Bias_G
:= \int_G |Q_{ \leq 2 m_0}f(g) - f(g)|^2 dg
= \|Q_{ \leq 2 m_0}f(g) - f(g)\|^2_{L^2} = \|Q_{ \leq 2 m_0}^\perp f\|^2_{L^2} < \varepsilon
\end{equation*}
by Parseval type identity (in~\cref{Fact_Parseval}), $m_0 := \REPf +1$ and the definition of $\REPf$ (in~\cref{Def_REPf}).
These arguments show $B_{\varepsilon}(f) \leq 4 m_0 = 4(\REPf +1) = O(\REPf)$.
\end{proof}

\cref{Thm_tight_B} is a direct consequence of \cref{Prop_Lower_B,Prop_B_upper}.
\begin{Thm1}[Rephrased]
For any $f \in L^2(\Ud)$, $B_\varepsilon(f) = \Theta(\REPf)$.
\end{Thm1}

In the proof of~\cref{Prop_B_upper}, the construction of an unbiased estimator for any polynomial function $f \in \mathcal{Q}_{ \leq m}(\Ud)$  is provided. Of course, this estimator may be used for the framework of PAC learning. In~\cref{Prop_poly_est_PAC},  we show an upper bound on the query complexity of estimating a polynomial $f$ in PAC learning framework.

\begin{Prop1}
For a $f \in \mathcal{Q}_{\leq m}(\Ud)$,
$Q_{\varepsilon, \delta}(f)
= O\left(\frac{\|A\|_1^2 \log \frac{1}{\delta}}{\varepsilon^2 }\cdot m\right)$ for any matrix $A$ satisfying 
\begin{equation*}
f(g) =\Tr A\left(\bigoplus_{0 \leq n, n' \leq m} g^{\otimes n} \otimes g^{\ast \otimes n'} \otimes I_E\right).
\end{equation*}
\end{Prop1}

\begin{proof}
First recall a general strategy of estimation obtained by Hoeffding's inequality.
For $\mathbb{R}$-valued i.i.d. random variables $X_1, \ldots, X_n$ taking values on the interval $[a, b]$, Hoeffding's inequality ensures
\begin{equation*}
\Pr\left(|\overline{X} - \E[\bar{X}]| > t\right) \leq 2 \exp\left(\frac{-2t^2 n}{(b - a)^2}\right)
\end{equation*}
for any positive $t >0$, where $\overline{X} := \frac{X_1 + \cdots + X_n}{n}$.
Therefore setting $t = \varepsilon$ and $n > \frac{(b-a)^2}{\varepsilon^2} \log \frac{2}{\delta}$ yields
\begin{equation*}
\Pr\left(|\overline{X} - \E[\bar{X}]| > \varepsilon\right)
\leq 2 \exp \left(\frac{-2 \varepsilon^2 n}{(b - a)^2}\right)
< \delta.
\end{equation*}
Therefore, to estimate a true value with an unbiased estimator in the $(\varepsilon, \delta)$-PAC-learning framework, it suffices to repeat the estimation process for $n = O\left(\frac{(b -a)^2}{\varepsilon^2}\log \frac{1}{\delta} \right)$ times and take the average.
\par
For $\mathbb{C}$-valued random variables, we decompose $X_i = \Re X_i + i\Im X_i$ and obtain
\begin{align*}
\Pr\left( | \overline{\Re X} + i\overline{\Im X}| > \varepsilon \right)
&\leq \Pr\left( | \overline{\Re X}| + |\overline{\Im X}| > \varepsilon \right) \\
&\leq \Pr\left( | \overline{\Re X}| > \varepsilon/2 \right) + \Pr\left( |\overline{\Im X}| > \varepsilon/2 \right) \\
&\leq 4 \exp\left(\frac{-\varepsilon^2 n}{8C_0^2}\right)
\end{align*}
from the same argument with Hoeffding's inequality, 
where $C_0$ is the maximum possible values of $|X_i|$.
Therefore, it suffices to repeat the estimation process for $n = O\left(\frac{C_0^2}{\varepsilon^2}\log \frac{1}{\delta} \right)$ times and take the average.

\par
We now apply this argument to our estimation scheme that is essentially given in the proof of~\cref{Prop_B_upper}. 
In our estimation scheme, first observe the function $f$ may be expressed as
\begin{equation*}
f(g) = \|A\|_1 \cdot \sum_{i} \frac{\sigma_i}{\|A\|_1} \langle{e_i}|V \left(\bigoplus_{0 \leq n, n' \leq m} g^{\otimes n} \otimes g^{\ast \otimes n'} \otimes I_E\right)U \Ket{e_i},
\end{equation*}
where the singular value decomposition $A = UDV$ is applied.
And then as in the second part of~\cref{Prop_B_upper}, apply $\Est{V^\ast e_i}{U e_i}$ with probability $\frac{\sigma_i}{\|A\|_1}$ (Note that $\sum_i \frac{\sigma_i}{\|A\|_1} = 1$) and multiply by $\|A\|_1$.
This process requires $m$ queries to the $\textrm{C-}g$ operation, and yields an unbiased estimator for $f(g)$ that takes values on the interval $[-\|A\|_1, \|A\|_1]$, because
\begin{equation*}
|\langle{e_i}|V\Ket{e_i}| \leq \|e_i\|_2 \|V e_i\|_2 = 1
\end{equation*}
for any unitary operator $V$ by Cauchy--Schwarz inequality.
\par
Therefore from the above discussion we obtain
$Q_{\varepsilon, \delta}(f)
= O\left(\frac{\|A\|_1^2}{\varepsilon^2}\log \frac{1}{\delta}\cdot m\right)$.
\end{proof}

\section{Applications}\label{sec_Appli}
\cref{sec_Appli} has two subsections: \cref{subsec_Simp} and \cref{subsec_Opt}. 
\cref{subsec_Simp} discusses how to obtain simpler expressions of the quantity $\REPf$, and \cref{subsec_Opt} shows that our algorithm in fact works well even in PAC learning framework for specific functions.
\subsection{Simpler forms of $\REPf$}\label{subsec_Simp}
Here we aim to derive simpler forms of the quantity $\REPf$ for specific functions such as univariate polynomials, the trace function, and matrix elements of unitary irreducible representations, on $\Ud$.

\subsubsection{Univariate Polynomials}
Define 
$G(\alpha, d) := \int |g_{11}|^{2\alpha} dg$
which can be computed by properties of the Haar measure on the unitary group~\cite[Proposition~2.5]{Mec19}. For example,
\begin{equation*}
G(\alpha = 1, d) = \int |g_{11}|^{2} dg = \frac{1}{d},
\end{equation*}
since $g_{11}, \ldots, g_{1d}$ are i.i.d., and $\sum_{1 \leq i \leq d} |g_{1i}|^2 = 1$.
\begin{Prop}
Let $\alpha \in \mathbb{N}_0$ and $f(g) = g_{11}^\alpha$.
Then
\begin{equation*}
\REPf = 
\begin{cases}
\lfloor \frac{\alpha -1}{2} \rfloor &\text{if $\varepsilon < G(\alpha, d)$,}\\
0&\text{otherwise.}
\end{cases}
\end{equation*}
\end{Prop}
\begin{proof}
By defition of $G(\alpha, d)$, we have
\begin{equation*}
G(\alpha, d) = \|Q^\perp_{\leq (\alpha - 1)}f\|^2_{L^2} >  \|Q^\perp_{\leq \alpha}f\|^2_{L^2} = 0,
\end{equation*}
together with the fact that $f$ is an $\alpha$-degree polynomial and the definition of the projection $Q^\perp_{\leq \alpha}$.
Since 
\begin{equation*}
2\left\lfloor \frac{d-1}{2} \right\rfloor < d \leq 2\left(\left\lfloor \frac{d-1}{2} \right\rfloor + 1\right)
\end{equation*}
 for any $d \in \mathbb{N}_0$,
this shows the statement.

\end{proof}
As shown in~\cref{Lem_poly_decomp}, for a different $\alpha \in \mathbb{N}_0$, the function $g^\alpha_{11}$ belongs to a different orthogonal subspace $\mathcal{P}_\alpha(\Ud)$. Therefore the $L^2$ norm of a univariate polynomial $f = \sum_{0 \leq \alpha' \leq \alpha} a_{\alpha'} g_{11}^{\alpha'}$ satisfies
\begin{equation*}
\|f\|_{L^2}^2
= \sum_{0 \leq \alpha' \leq \alpha} |a_{\alpha'}|^2 \|g_{11}^{\alpha'} \|^2_{L^2}
= \sum_{0 \leq \alpha' \leq \alpha} |a_{\alpha'}|^2 G(\alpha', d).
\end{equation*}
This yields~\cref{Cor_univariate}.

\begin{Cor}\label{Cor_univariate}
For a degree $\alpha$, univariate polynomial $f(g) = \sum_{0 \leq \alpha' \leq \alpha} a_{\alpha'} g_{11}^{\alpha'}$,
\begin{equation*}
\REPf = \max \left\{m \mid \sum_{\alpha' > 2m} |a_{\alpha'}|^2 G(\alpha', d) \geq \varepsilon \right\}.
\end{equation*}
\end{Cor}

\subsubsection{The trace function}
Since the (normalized) trace function $f(g) = \frac{1}{d}\Tr g $ belongs to $\mathcal{P}_1(\Ud)$, we have~\cref{Fact_REP_trace}.
(The normalization factor does not provide any role for~\cref{Fact_REP_trace}, which plays a role rather in~\cref{Fact_PAC_trace}.)

\begin{Fact}\label{Fact_REP_trace}
When $f(g) = \frac{1}{d}\Tr g$, $\REPf = 0$.
\end{Fact}

\begin{Remark}
Even though $B_\varepsilon(f) = \Theta(\REPf)$ holds by~\cref{Thm_tight_B}, \cref{Fact_REP_trace} does not imply $B_\varepsilon(f) = 0$ because the $\Theta$ notation hides an additive constant factor.
\end{Remark}

\subsubsection{The determinant $\textrm{det}~g$}
By Shur's orthogonality relation~\cite[Section~4]{BT03}, the one dimensional irreducible representation $(\textrm{det} g, \mathbb{C})$ satisfies $\|\det g\|^2_{L^2} = 1$.
Since $\det g$ is a polynomial with degree $d$ by definition, we obtain~\cref{Fact_REP_det}.

\begin{Fact}\label{Fact_REP_det}
When $f(g) = \det~g$, $\REPf = \lfloor \frac{d - 1}{2}\rfloor$.
\end{Fact}

More generally, any irreducible representation $(\pi_\lambda, V_\lambda)$ over $\Ud$ satisfies $\|\pi_\lambda(g)_{i, j}\|^2_{L^2} = 1/\textrm{dim} V_\lambda$.
In addition, any label $\lambda$ must belong to some set $\Lambda_d(m, \bar{m})$ from which Mixed Schur--Weyl duality, in~\cref{Fact_Mixed-Schur-Weyl}, tells that $\pi_\lambda (g)_{i, j}$ is a linear combination of polynomials with degree exactly equal to $m + \bar{m}$.
These arguments show~\cref{Fact_REP_irreps}.

\begin{Fact}\label{Fact_REP_irreps}
For any unitary irreducible representation~$(\pi_\lambda, V_\lambda)$ whose label $\lambda$ belongs to a set $\Lambda_d(m, \bar{m})$,
\begin{equation*}
\REPf = 
\begin{cases}
 \left\lfloor \frac{m + \bar{m} - 1}{2} \right\rfloor &\varepsilon < 1/\textrm{dim} V_\lambda,\\
 0 &\text{otherwise.}
\end{cases}
\end{equation*}
\end{Fact}

\subsection{Optimal PAC estimations}\label{subsec_Opt}

In~\cref{subsec_Opt}, we apply~\cref{Prop_poly_est_PAC} and obtain upper bounds on the PAC-learning complexity $Q_{\varepsilon, \delta}(f)$ for several different functions.
As observed in~\cref{subsec_Simp} together with~\cref{Fact_B_Q}, these upper bounds are tight for sufficiently small constants $\varepsilon, \delta$.

\subsubsection{Univariate Polynomials}
\begin{Prop}\label{Prop_Uni_variate}
Let $\alpha \in \mathbb{N}_0$ and $f(g) = g_{11}^\alpha$.
Then $Q_{\varepsilon, \delta}(f) = O(\frac{\alpha}{\varepsilon^2}\log \frac{1}{\delta})$
\end{Prop}
\begin{proof}
First notice that there is a pair $(i, j)$ such that the $(i, j)$ element of the matrix
$\bigoplus_{0 \leq n, n' \leq \alpha} g^{\otimes n} \otimes g^{\ast \otimes n'}$ is equal to $g_{11}^\alpha$.
We now apply \cref{Prop_poly_est_PAC} with $E := \mathbb{C}$ and the matrix $A := E_{ij}$, whose $(i, j)$ element is one, and all the other elements are zero.
We then have 
\begin{equation*}
g^\alpha_{11} =\Tr A\left(\bigoplus_{0 \leq n, n' \leq m} g^{\otimes n} \otimes g^{\ast \otimes n'} \otimes I_E\right).
\end{equation*}
Since $g^\alpha_{11} \in \mathcal{Q}_{\leq \alpha}(\Ud)$ and $\|A\|_1 = \Tr\sqrt{E_{ij}^\ast E_{ij}} = 1$, we have $Q_{\varepsilon, \delta}(f) = O(\frac{\alpha}{\varepsilon^2} \log \frac{1}{\delta})$.
\end{proof}

\subsubsection{The trace function}

\begin{Fact}\label{Fact_PAC_trace}
When $f(g) = \frac{1}{d}\Tr g$, $Q_{\varepsilon, \delta}(f) = O(\frac{1}{\varepsilon^2} \log \frac{1}{\delta})$.
\end{Fact}

\begin{proof}
First observe that $\bigoplus_{0 \leq n, n' \leq 1} g^{\otimes n} \otimes g^{\ast \otimes n'} = g \otimes I \oplus I \otimes g^\ast \oplus g \otimes g^\ast$.
Therefore define $A := \frac{1}{d}(I \otimes E_{11}) \oplus 0 \oplus 0$ and we obtain
\begin{equation*}
\Tr A \bigoplus_{0 \leq n, n' \leq 1} g^{\otimes n} \otimes g^{\ast \otimes n'} = \frac{1}{d} \Tr g \otimes E_{11} = \frac{1}{d} \Tr g.
\end{equation*}
Since $\|A\|_1 = 1$, we obtain the desired statement.
\end{proof}

\subsubsection{Irreducible representations $\pi_\lambda(g)_{i, j}$}

\begin{Prop}\label{Prop_PAC_irreps}
For any unitary irreducible representation~$(\pi_\lambda, V_\lambda)$ whose label $\lambda$ belongs to a set $\Lambda_d(m, \bar{m})$,
\begin{equation*}
Q_{\varepsilon, \delta}(f) = O\left(\frac{m + \overline{m}}{\varepsilon^2} \log \frac{1}{\delta}\right)
\end{equation*}
where $f(g) = \pi_\lambda(g)_{i, j}$.
\end{Prop}

\begin{proof}
For any $\lambda_0 = (\lambda_+, \lambda_-) \in \Lambda_d(m, \overline{m})$, there are unitaries $U, V$ and $W$ such that
\begin{align*}
U^\ast g^{\otimes m} U &= \bigoplus_{\lambda \in \Lambda_{m, d}} \pi_\lambda(g) \otimes I = \pi_{\lambda_+}(g) \oplus \Garb, \\
V^\ast \bar{g}^{\otimes m} V &= \bigoplus_{\overline{\lambda} \in \bar{\Lambda}_{m, d}} \pi_{\bar{\lambda}}(g) \otimes I = \pi_{\lambda_+}(g) \oplus \Garb, \\
W^\ast \pi_{\lambda_+}(g) \otimes \pi_{\lambda_-}(g) W &= \pi_{(\lambda_+, \lambda_-)}(g) \oplus \Garb
\end{align*}
by Schur--Weyl duality, where the last expression immediately comes from the highest weight theory.
Therefore for any quantum system $E$,
\begin{equation}\label{eq_Prop_PAC_irreps}
\left[U^\ast \otimes V^\ast g^{\otimes m} \otimes \overline{g}^{\otimes \bar{m}} U \otimes V \right]\otimes I_E
= W [\pi_{(\lambda_+, \lambda_-)}(g) \oplus \Garb_1] W^\ast \oplus \Garb_2.
\end{equation}
Now let $E := \mathbb{C}^{\dim (\pi_{\lambda_+} \otimes \pi_{\lambda_-})}$ and take the partial transpose over the systems $BE$, where the system $B$ corresponds to the space on which $\bar{g}^{\otimes \bar{m}}$ of the RHS acts.
This changes Equation~\eqref{eq_Prop_PAC_irreps} into
\begin{equation*}
\left[U^\ast \otimes \bar{V} g^{\otimes m} \otimes g^{\ast \otimes \bar{m}} U \otimes V^\T \right]\otimes I_E
= \mqty( \overline{W} [\pi^\T_{(\lambda_+, \lambda_-)}(g) \oplus \Garb_1] W^\T & \Garb_a\\
\Garb_b & \Garb_c
)
\end{equation*}
from~\cref{Lem_partial}, due to the definition of $E$.
Therefore taking $\tilde{A}:= \overline{W} [E_{ij} \oplus 0_{\Garb_1}]W^\T \oplus 0_{\Garb_2}$, we have
\begin{align*}
\Tr \tilde{A}\left[U^\ast \otimes \bar{V} g^{\otimes m} \otimes g^{\ast \otimes \bar{m}} U \otimes V^\T \right]\otimes I_E
&= \Tr\left[ \tilde{A} 
\mqty( \overline{W} [\pi^\T_{(\lambda_+, \lambda_-)}(g) \oplus \Garb_1] W^\T & \Garb_a\\
\Garb_b & \Garb_c
)\right]\\
&= \pi_{\lambda_0}(g)_{i, j}.
\end{align*}
Finally let $A = \left(U \otimes V^\T \tilde{A} U^\ast \otimes \bar{V}\right) \oplus 0$ on the space for $\bigoplus_{0 \leq n, n' \leq m + \bar{m}} g^{\otimes n} \otimes g^{\ast \otimes n'} \otimes I_E$, 
\begin{align*}
 \Tr A \bigoplus_{0 \leq n, n' \leq m + \bar{m}} g^{\otimes n} \otimes g^{\ast \otimes n'} \otimes I_E
&= \Tr\left[ \tilde{A} 
\mqty( \overline{W} [\pi^\T_{(\lambda_+, \lambda_-)}(g) \oplus \Garb_1] W^\T & \Garb_a\\
\Garb_b & \Garb_c
)\right]\\
&= \pi_{\lambda_0}(g)_{i, j}.
\end{align*}
Since $\|A\|_1 =\|U \otimes V^\T \tilde{A} U^\ast \otimes \bar{V}\|_1 = \|\tilde{A}\|_1 =  1$, we obtain the desired statement.
\end{proof}

\section*{Acknowledgement}
The author would like to thank his boss Michał Oszmaniec for his kindness and support.
The author acknowledges support from National Science Center, Poland
within the QuantERA III Programme (No 2023/05/Y/ST2/00140 acronym Tuquan) that
has received funding from the European Union's Horizon 2020 program.

\bibliography{
/Users/daikisuruga/Dropbox/Citations/mathematics_D, 
/Users/daikisuruga/Dropbox/Citations/computational_D, 
/Users/daikisuruga/Dropbox/Citations/query_D, 
/Users/daikisuruga/Dropbox/Citations/quant_info_D, 
/Users/daikisuruga/Dropbox/Citations/comm_comp_D, 
/Users/daikisuruga//Dropbox/Citations/books_D,
/Users/daikisuruga//Dropbox/Citations/tomography_D
}
\bibliographystyle{alpha}

\appendix
\section{Appendix}\label{sec_App}
Here we give proofs of several statements that have to place in Appendix.

For any $m \in \mathbb{N}_0$, define a subspace of $L^2(\Ud)$ as 
\begin{equation*}
\mathcal{P}_m(\Ud)
:= \mathrm{span}_\mathbb{C}\left\{ g_{11}^{x_{11}}g_{12}^{x_{12}} \cdots g_{dd}^{x_{dd}} \mid x_{11} + x_{12} + \cdots x_{dd} = m \right\}
\end{equation*}
where $g_{ij}$ denotes the $(i, j)$ entry of $g \in \Ud$. This is the space of multi-variate polynomials in matrix entry of $g$, whose degree is at most $m$. 
We also define another subspace $\overline{\mathcal{P}}_m(\Ud)$ as
\begin{equation*}
\overline{\mathcal{P}}_m(\Ud)
:= \mathrm{span}_\mathbb{C}\left\{ \bar{g}_{11}^{x_{11}}\bar{g}_{12}^{x_{12}} \cdots \bar{g}_{dd}^{x_{dd}} \mid x_{11} + x_{12} + \cdots x_{dd} = m \right\}.
\end{equation*}
The subspaces $\mathcal{P}_m(\Ud), \overline{\mathcal{P}}_m(\Ud)~ (m \in \mathbb{N}_0)$ are finite dimensional, and therefore closed in $L^2(\Ud)$. 
Note that $\mathcal{P}_0(\Ud) = \overline{\mathcal{P}}_0(\Ud) = \mathbb{C}$.
Analogously, define 
\begin{equation*}
\mathcal{P}_{\leq m}(\Ud)
:= 
\bigoplus_{0 \leq m' \leq m} \mathcal{P}_{m'}(\Ud)
\text{~and~}
\overline{\mathcal{P}}_{\leq m}(\Ud)
:= 
\bigoplus_{1 \leq m' \leq m} \overline{\mathcal{P}}_{m'}(\Ud)
\end{equation*}
that are also finite-dimensional and therefore closed. Note that for convenience $\overline{\mathcal{P}}_0$ is not included in case of $\overline{\mathcal{P}}$.
\begin{Lem}\label{Lem_poly_decomp}

\begin{equation*}
\bigoplus_{\lambda \in \Lambda_{m, d}} M_{\pi_\lambda}
= \mathcal{P}_m(\Ud)
\end{equation*}
\end{Lem}

\begin{proof}
By~\cref{Fact_Schur-Weyl}, 
\begin{equation}\label{eq_poly_decomposition}
U_\mathsf{Sch} g^{\otimes m} U^\ast_\mathsf{Sch}
=
\mqty(\pi_{\lambda_1}^\Ud(g) \otimes I_{\dim \pi^\Sm_{\lambda_1}} &\mathbf{0}& \cdots & \mathbf{0}\\
\mathbf{0} & \pi_{\lambda_2}^\Ud(g) \otimes I_{\dim \pi^\Sm_{\lambda_2}} & \cdots & \mathbf{0}\\
\vdots &  & \ddots & \vdots \\
\mathbf{0} & \mathbf{0} & \mathbf{0} & \pi_{\lambda_{\max}}^\Ud(g) \otimes I_{\dim \pi^\Sm_{\lambda_{\max}}}\\
)
\end{equation}
where $\{\lambda_1, \ldots, \lambda_{\max}\} = \Lambda_{m, \leq d}$.
By the definition of tensor products, matrix entries of $g^{\otimes m}$ form a basis of $\mathcal{P}_m(\Ud)$, and therefore
Equation~\eqref{eq_poly_decomposition} implies any $\pi_\lambda(g)^\Ud~(\lambda \in \Lambda_{m, \leq d})$ may be written as a linear combination of elements in $\mathcal{P}_m(\Ud)$. 
This shows
\begin{equation*}
\bigoplus_{\lambda \in \Lambda_{m, \leq d}} M_{\pi_\lambda}
\subseteq \mathcal{P}_m(\Ud).
\end{equation*}
To show the opposite direction, use
\begin{equation}
 g^{\otimes m} =
U^\ast_\mathsf{Sch}
\mqty(\pi_{\lambda_1}^\Ud(g) \otimes I_{\dim \pi^\Sm_{\lambda_1}} &\mathbf{0}& \cdots & \mathbf{0}\\
\mathbf{0} & \pi_{\lambda_2}^\Ud(g) \otimes I_{\dim \pi^\Sm_{\lambda_2}} & \cdots & \mathbf{0}\\
\vdots &  & \ddots & \vdots \\
\mathbf{0} & \mathbf{0} & \mathbf{0} & \pi_{\lambda_{\max}}^\Ud(g) \otimes I_{\dim \pi^\Sm_{\lambda_{\max}}}\\
)U_\mathsf{Sch}
\end{equation}
and follow the same proof strategy.
\end{proof}

\begin{Lem}

\begin{equation*}
\bigoplus_{\lambda \in \overline{\Lambda}_{m, \leq d}} M_{\pi_\lambda}
= \overline{\mathcal{P}}_m(\Ud)
\end{equation*}
where 
\begin{equation*}
\overline{\Lambda}_{m, d} := \{\bar{\lambda} \mid \lambda \in \Lambda_{m, d}\} 
= \left\{\lambda \in \mathbb{Z}_+^d \mid \sum_{i \leq d} \lambda_i = -m, \lambda_1 \leq 0\right\}
\end{equation*}
\end{Lem}
\begin{proof}
Take the complex conjugation on both sides of Equation~\eqref{eq_poly_decomposition} and apply the same proof technique as~\cref{Lem_poly_decomp} together with~\cref{Fact_conjugate_label_Ud}.
\end{proof}

\begin{Lem}\label{Lem_partial}
Let $A, B, C$ and $D$ be rectangular matrices satisfying $A \otimes B = C \oplus D$. Then taking the partial transpose on B of the RHS yields
\begin{equation*}
A \otimes B^\T = 
\mqty(C^\T  & D_1\\
D_2 & D_3
)
\end{equation*}
for some $D_i$'s,
when $\dim B \geq \dim C$.
\end{Lem}

\begin{proof}
Let $X \in \mathbf{M}(n\times m, n \times m, \mathbb{C})$ be a matrix on the space $V \otimes W$ where $\dim V = n$ and $\dim W = m$.
First observe the partial transpose of X on the space $W$ is represented as
\begin{equation*}
\mqty(
X_{11}^\T & X_{12}^\T & \cdots & X_{1n}^\T \\
X_{21}^\T  & X_{22}^\T  & \cdots & X_{2n}^\T \\
\vdots &  & \ddots & \vdots\\
X_{n1}^\T & X_{n2}^\T &  & X_{nn}^\T
),
\end{equation*}
where $X = (X_{ij})_{1 \leq i, j \leq n}$ is a block decomposition; each $X_{ij}$ is an $m\times m$ matrix. 
Therefore the partial transpose of the matrix $C \oplus D = \mqty(C & 0\\ 0 & D)$ on the space whose dimension larger than or equal to $\dim C$ satisfies
\begin{equation*}
\mqty(C^\T  & D_1\\
D_2 & D_3
).
\end{equation*}
This shows the statement.

\end{proof}

\end{document}